%% file: geobidi.tex
\documentclass[runningheads,a4paper,wide,11pt]{article}

\usepackage[russian,spanish,greek,english]{babel}
\usepackage[utf8x]{inputenc}

\usepackage{amsmath,amsthm,amsfonts,amssymb,epsfig,color}

\newcommand{\noindents}[2]{\vspace{-#1mm}\noindent {#2}}

\newcommand{\tw}{{\mathbf{tw}}}

\newtheorem{corollary}{Corollary}
\newtheorem{theorem}{Theorem}
\newtheorem{lemma}{Lemma}

\newtheorem{definition}{Definition}
\newtheorem{proposition}{Proposition}

\newcommand{\bb}{{\cal B}}

\usepackage{titling}

\begin{document}
\thanksmarkseries{arabic}

\title{Bidimensionality of Geometric Intersection Graphs\thanks{Emails:  Alexander Grigoriev: {\sf a.grigoriev@maastrichtuniversity.nl},
Athanassios Koutsonas: {\sf akoutson@math.uoa.gr}, 
Dimitrios  M. Thilikos: {\sf sedthilk@thilikos.info}}}

\author{Alexander Grigoriev\thanks{School of Business and Economics Department of Quantitative Economics, Maastricht University, Maastricht.}
\and Athanassios Koutsonas\thanks{Department of Mathematics, National and Kapodistrian University of Athens, Athens.} 
\and  Dimitrios M. Thilikos$^{3,}$\thanks{AlGCo project-team, CNRS, LIRMM.}\hspace{1.4mm}$^{,}$\hspace{.2mm}\thanks{Co-financed by the European Union (European Social Fund -- ESF) and
Greek national funds through the Operational Program ``Education and Lifelong Learning'' of the
National Strategic Reference Framework (NSRF) - Research Funding Program:
``{\sl Thalis. Investing in knowledge society through the European Social Fund}''.}}
\date{\today}

\maketitle

\begin{abstract}
\noindent Let ${\cal B}$ be a finite collection of geometric (not necessarily convex) bodies in the plane. Clearly, this class of geometric objects naturally generalizes the class of disks, lines, ellipsoids, and even convex polygons. We consider geometric intersection graphs $G_{\cal B}$ where each body of the collection ${\cal B}$ is represented by a vertex, and two vertices of $G_{\cal B}$ are adjacent if the intersection of the corresponding bodies is non-empty. 
For such graph classes and under natural restrictions 
on their maximum degree or subgraph exclusion,
we prove that the relation between their treewidth and the maximum size of a grid minor
is linear.
These combinatorial results vastly  extend the applicability of
all the  meta-algorithmic results of the bidimensionality theory to geometrically defined graph classes.

\end{abstract}


\section{Introduction}\label{intro}

Parameterized complexity treats problems as subsets of ${\rm \Sigma}^{*}\times \Bbb{N}$, for some alphabet ${\rm \Sigma}$. An instance of a parameterized problem is a pair $(I,k)$ where $I$ is the main part of the problem description and $k$ is a, typically small, parameter.  An {\em {\sf FPT} algorithm} for a parameterized problem ${\rm {\rm \Pi}}$ is one that runs in $f(k)\cdot n^{O(1)}$ time. A central issue in parameterized complexity is to find which parameterized problems admit {\sf FPT} algorithms and, when this is the case, to reduce as much as possible the contribution of the function $f(\cdot)$, i.e., their {\em parametric dependance}. {\sf FPT} algorithms where $f(k)=2^{o(k)}$ are called {\em sub-exponential parameterized algorithms}. It is known that such an algorithm where  $f(k)=2^{o(\sqrt{k})}$ is unlikely to exist for several problems on graphs, even when restricted to sparse graph classes such as planar graphs~\cite{CaiJ03onth}. Therefore, a parametric dependance $f(k)=2^{O(\sqrt{k})}$ is the best we may expect and this is what we may aim for.

A kernelization algorithm for a parameterized problem ${\rm {\rm \Pi}}$ is one that, in polynomial time, can
replace any instance $(I,k)$ with a new equivalent one whose size depends exclusively on the parameter $k$. If such an algorithm exists and the size of the new instance is linear in $k$, then we say that ${\rm {\rm \Pi}}$ {\em admits a linear kernel}. While the existence of an {\sf FTP} algorithm implies the existence of a kernel it is a challenge to find for which problems such a kernel can be polynomial~\cite{BodlaenderDFH09onpr}.

\paragraph{\bf Bidimensionality theory.} This theory
was initially introduced in~\cite{DemaineFHT05sube}
as a general framework for designing parameterized algorithms  with sub-exponential parametric dependance.
Moreover, it also provided meta-algo\-rithmic results 
in approximation algorithms~\cite{DemaineHaj05bidi,FominLRS11bidi} and 
kernelization~\cite{FominLST10Bidi}  (for  a survey on bidimensionality, see~\cite{DemaineH07-CJ}).
To present the consequences and the motivation of our results let us first give some brief description of 
the meta-algorithmic consequences of Bidimensionality 
theory. 
For this, we need first some definitions.

A {\em graph invariant} is a function ${\bf p}$ mapping graphs to non-negative 
integers. The parameterized problem associated to an invariant  ${\bf p}$
has as input a pair $(G,k)$ where $G$ is a graph and $k$ is a non-negative integer, and 
asks whether ${\bf p}(G)\leq k$ (or, alternatively whether ${\bf p}(G)\geq k$).
Let ${\cal G}$ be the set of all graphs. 
The parameterized problem corresponding to ${\bf p}$ is denoted by ${\rm {\rm \Pi}}_{\bf p}\subseteq 
{\cal G}\times \Bbb{N}$ and is defined as ${\rm {\rm \Pi}}_{\bf p}=\{(G,k)\mid {\bf p}(G)\leq k\}$ or, alternatively, as
${\rm {\rm \Pi}}_{\bf p}=\{(G,k)\mid {\bf p}(G)\geq k\}$.
We also define the graph invariant ${\bf bg}$ such that given a graph $G$,  
$${\bf bg}(G)=\max\{k\mid G\mbox{~contains the $(k\times k)$-grid as a minor}\}.$$
\begin{definition}
Given a {\em graph invariant} ${\bf p}$ we say that ${\rm \Pi}_{\bf p}$ is {\em minor-bidimensional}
if the following conditions hold:
\begin{itemize}
\item ${\bf p}$ is closed under taking of subgraphs, i.e., for every $G\in{\cal G}$, if $H$ is a minor of $G$, then ${\bf p}(H)\leq {\bf p}(G)$.
\item If $L_{k}$ is the $(k\times k)$-grid, then ${\bf p}(L_{k})={\rm \Omega}(k^2)$.
\end{itemize}
\end{definition}

\noindent The main consequences of bidimensionality theory for minor closed invariants are summarized by the following:
Suppose that ${\bf p}$ is a graph invariant such that ${\rm \Pi}_{\bf p}$ 
is a minor bidimensional problem. 
Let also ${\cal G}$ be a graph class
such that it satisfies the following property:
\begin{eqnarray}
\forall_{G\in {\cal G}}\ \tw(G) & =& O({\bf bg}(G))\label{gjjfkr43d}
\end{eqnarray}
and let $\Pi_{\bf p}^{\cal G}$ be the restriction of ${\rm \Pi}_{\bf p}$ to 
the graphs in ${\cal G}$, i.e. the  problem occurring if we alter all YES-instance of
$\Pi_{\bf p}^{\cal G}$ whose graph is  not in ${\cal G}$ to NO-instances. Then the following hold

\begin{enumerate}
\item if ${\bf p}(G)$ can be computed in $2^{O(\tw(G))}\cdot n^{O(1)}$ steps, then 
${\rm \Pi}_{\bf p}^{\cal G}$ can be solved by a sub-exponential parameterized algorithm that runs in $2^{O(\sqrt{k})}\cdot n^{O(1)}$ steps. 
\item if ${\bf p}$  satisfies some separability property (see~\cite{FominLRS11bidi,DemaineHaj05bidi,FominLST10Bidi} for the precise definition) and ${\rm \Pi}_{\bf p}=\{(G,k)\mid \exists S\subseteq V(G): 
|S|\geq k\mbox{~and~} (G,S)\models \psi\}$ where $\psi$  is a sentence in Counting Monadic Second Order logic, then  ${\rm \Pi}_{\rm p}$ admits a linear kernel, i.e. there exists a polynomial algorithm reducing 
each instance $(G,k)$ of  ${\rm \Pi}_{\rm p}$ to an equivalence instance $(G',k')$ where $|V(G')|=O(k)$
and $k'\leq k$.
\item If ${\bf p}$ satisfies some separability property and is reducible (in the sense this is defined in~\cite{FominLRS11bidi}), then there is an EPTAS for computing ${\bf p}(G)$ on the graphs in ${\cal G}$.
\end{enumerate}

\noindent  According the the current state of the art all above meta-algorithmic results hold when  ${\cal G}$
excludes graphs with some fixed  graph $H$ as a minor. This is due to the combinatorial result of 
Demaine and Hajiaghayi  in~\cite{DemaineH08line}, who proved~\eqref{gjjfkr43d} 
for every graph $G$ excluding some fixed graph $H$ as a minor.
While such graphs are of a topological nature
it remained an interesting question whether the applicability 
of the above theory can be extended for geometrically (rather than topologically)
restricted graphs classes. 

%

\paragraph{\bf Our results.} Clearly, any extension of the applicability of bidimensionality theory 
on some class ${\cal G}$ requires a proof that it satisfies  property~\eqref{gjjfkr43d}.
Recently, a first step to extend meta-algorithmic results for graph classes that are not topologically restricted was done in~\cite{FominLS12bidi}, where the bidimensionality condition was used to derive sub-exponential algorithms for $H$-free unit-disk intersection graphs and $H$-free map graphs, where a graph class is $H$-free if none of its graphs contains $H$ as a subgraph. However, no meta-algorithmic results were known so far for more generic classes of geometric intersection graphs, like e.g. intersection graphs of polygonal objects in the plane. 
In this paper we vastly extend the combinatorial results of~\cite{FominLS12bidi}
to more general families of geometric intersection graphs. In particular, we prove that 
property~\eqref{gjjfkr43d} holds for several classes 
of geometric intersection graphs and open 
a new direction of the applicability of bidimensionality theory. 
In particular our results are the following.

\begin{itemize}
\item[1.] Let ${\cal B}$ be a set of (not necessarily straight) lines in the plane such that for each $C_{1},C_{2}\in {\cal B}$ with $C_{1}\neq C_{2}$, the set $C_{1}\cap C_{2}$ is a finite set of points and at most two lines intersect in the same point. 
Assume also that each line is intersected at most $\xi$ times. Then $\tw(G_{\cal B})=O(\xi\cdot {\bf bg}(G_{\cal B}))$.\\


\item[2.] Let ${\cal B}$ be a set of $\rho$-convex bodies (bodies where any two of their points can be joined by a polysegment of at most $\rho-1$ bends that is entirely inside the body) such that for each $B_{1},B_{2}\in {\cal B}$ with $B_{1}\neq B_{2}$, the set $B_{1}\cap B_{2}$ has a non-empty interior. Let $G_{\cal B}$ be the intersection graph of ${\cal B}$ and let $\Delta$ be the maximum degree of $G_{\cal B}$. Then $\tw(G_{\cal B})=O(\rho^2\Delta^3\cdot {\bf bg}(G_{\cal B}))$.\\

\item[3.] Let $H$ be a graph on $h$ vertices, and let $\bb$ be a collection of convex bodies in the plane such that for each $B_{1},B_{2}\in {\cal B}$ with $B_{1}\neq B_{2}$, the set $B_{1}\cap B_{2}$ has a non-empty interior. If the intersection graph $G_{\bb}$ of $\bb$ is $\alpha$-fat and does not contain $H$ as a subgraph, then $\tw(G_{\cal B})=O(\alpha^6 h^3\cdot {\bf bg}(G_{\cal B}))$. (Given a real number $\alpha $, we call the intersection graph of a collection of convex bodies
{\em $\alpha $-fat} if the ratio between the maximum and the minimum radius of a circle 
where all bodies of the collection can be circumscribed, and inscribed
respectively,  is upper bounded by $\alpha$.)
\end{itemize}

\noindent  Notice that the case of $H$-subgraph free unit-disk intersection graphs treated in~\cite{FominLS12bidi} is just a very special case of the fourth result (unit-disk graphs
are 1-convex and $1$-fat).\\

The paper is organized as follows: In Section~\ref{defprel},
we give some basic definitions and results.
In Section~\ref{ldfmnfd} we prove the main technical results that are used
in Section~\ref{kfofhfm} for the derivation of its implications in a variety of geometric graph classes.
Section~\ref{kdhufkl} discusses extensions and conclusions of this work.
Al proofs have been moved to the appendix except from those of Lemmata 6, 7 and Theorem~1.

\section{Definitions and preliminaries}\label{defprel}

All graphs in this paper are undirected and may have loops or multiple edges. If a graph has no multiple edges or loops we call it {\em simple}. Given a graph $G$, we denote by $V(G)$ its vertex set and by $E(G)$ its edge set. Let $x$ be a vertex or an edge of a graph $G$ and likewise for $y$; their distance in $G$, denoted by ${\bf dist}_{G}(x,y)$ is the smallest length of a path in $G$ that contains them both. We call {\em part of a path} any sequence of adjacent edges in a given path. For any set of vertices $S\subseteq V(G)$, we denote by $G[S]$ the subgraph of $G$ induced by the vertices from $S$.

\noindents{-2}{\bf  Graph embeddings}.
 We use the term {\em graph embedding} to denote a drawing of a graph $G$ in the plane, where each vertex is associated to a distinct point of the plane and each edge to a simple open Jordan curve, such that its endpoints are the two points of the plane associated with the endvertices of this edge. To simplify the presentation, when not necessary, we do
not distinguish between a vertex of $G$ and the point in the plane representing the vertex; likewise for an edge of $G$. Roughly speaking, we often do not distinguish between $G$ and its embedding.
Two edges of an embedding of a graph in the plane {\em cross}, if they share a non-vertex point of the plane.
We use the term {\em plane graph} for an embedding of a graph
without crossings.
A graph is {\em planar} if it admits a plane embedding.

\noindents{-2}{\bf Geometric bodies, lines and polysegments}. We call a set of points in the plane a {\em 2-dimensional geometric body}, or simply a {\em 2-dimensional body}, if it is homeomorphic to the closed disk $\{ (x,y)|\ x^{2}+y^{2}\leq 1\}$. Also a {\em line} is 
a subset of the plane that is homeomorphic to the interval $[0,1]$.
 A {\em polysegment} $C$ is a line that is the union of a sequence of straight lines
$\overline{p_1p_2},\overline{p_2p_3},\cdots,\overline{p_{k-1}p_k}$ in the plane, where $p_1$ and $p_k$ are the endpoints of $C$. We say that a polysegment $C$ {\em contains} a point $p_i$ and {\em joins} the endpoints $p_1,\ p_k$, and we refer to the rest points $p_2,p_3,\cdots,p_{k-1}$ as {\em bend points} of $C$. The length of a polysegment is defined as equal to the number of straight lines it contains (i.e. one more than the number of its bend points). Throughout the paper we assume that a polysegment is not self-crossing.

\noindents{-2}{\bf Minors and distance minors}.
Given two graphs $H$ and $G$, we write $H\preccurlyeq G$ and call $H$ a {\em minor} of $G$, if $H$ can be obtained from a subgraph of $G$ by edge contractions
(the {\em contraction} of an edge $e=\{x,y\}$  in a graph $G$
 is the operation of replacing $x$ and $y$ by a new vertex $x_{e}$ that is made adjacent
with all the neighbors of $x$ and $y$ in $G$ that are different from $x$ and $y$).
Moreover, we say that $H$ is a {\em contraction} of $G$,
if $H$ can be obtained from $G$ by contracting edges.

Let $G$ be a simple graph. We denote as $G^{\rm \ell}$ the graph obtained from $G$ by adding a loop on each of its vertices.
 We also say that a subset $F$ of $E(G^{\rm \ell})$ is {\em solid},
 if for every $v_{1},v_{2}\in\bigcup_{e\in F}e$ there is a walk in $G^\ell$ from $v_{1}$ to $v_{2}$
 consisting of edges in $F$ and where each second edge is a loop.
We define the relation $\preccurlyeq_\phi$ between two graphs as follows.

 Let $H$ and $G$ be simple graphs.
Then we write $H\preccurlyeq_\phi G$, if there is a function  $\phi: E(G^{\rm \ell })\rightarrow V(H)\cup E(H)\cup\{\star\}$, such that
\begin{itemize}
\item[1.] for every vertex $v\in V(H)$, $\phi^{-1}(v)$ is a non-empty solid set,
\item[2.]
for every two distinct vertices $v_{1},v_{2}\in V(H)$, an edge in $\phi^{-1}(v_1)$ does not share a common endpoint with an edge in $\phi^{-1}(v_{2})$.
\item[3.]
for every edge  $e=\{v_1,v_2\}\in E(H)$ and every edge $e'$ in $\phi^{-1}(e)$, $e'$ is not a loop and  shares its one endpoint with an edge in $\phi^{-1}(v_1)$ and the other with an edge in $\phi^{-1}(v_2)$.

\item[4.] for every $e\in E(H)$, $|\phi^{-1}(e)|=1$.
\end{itemize}

The following lemma reveals the equivalence between the relation defined previously and the minor relation (for the proofs see Appendix). 

\begin{lemma}
\label{tldkgn}
If $G$  and  $H$ are graphs, then $H\preccurlyeq_\phi G$ if and only if $H$ is a minor of $G$.
\end{lemma}

\begin{proof}
Let $H\preccurlyeq_\phi G$. First notice, that by the definition of the function $\phi$ any loop of $G^\ell$ will be either discarded or mapped to a vertex of $H$. Conditions (1) and (2) guarantee that any two vertices $x,y$ of $H$ are corresponding to vertex-disjoint connected subgraphs $G_x,G_y$ of $G$. Moreover by (3), if $xy \in E(H)$ there is an edge in $G$ joining $G_x$ and $G_y$. Hence, $H$ can be obtained from a subgraph of $G$ by contracting the edges of the subgraphs $G_z$ ($z\in V(H)$) and thus is a minor of $G$.

Let now $H$ be a minor of $G$. This means that there exists a subgraph of $G$ consisting of disjoint trees $\{T_v\mid v\in V(H)\}$ plus a set of edges $E=E(H)$, such that contracting all the edges of the trees yields the graph $H$. Then, we choose $\phi$ as a function that maps the edges of a tree $T_x$ as well as the loops on the vertices of $T_x$ to $x$, the edges of $E$ to $E(H)$ and all other edges of $G^\ell$ to $\star$.
\end{proof}
%


Given the existence of a function $\phi$ as in the definition above, we say {\em $H$ is a $\phi$-generated minor of $G$}. Moreover, {\em $H$ is a distance minor of $G$} if $H$ is a $\phi$-generated minor of $G$ and the following additional condition holds:
%

\begin{itemize}
\item[5.] for every $e_{1},e_{2}\in E(G)\setminus \phi^{-1}(\star)$, ${\bf dist}_{H}(\phi(e_{1}),\phi(e_{2}))\leq {\bf dist}_{G}(e_{1},e_{2})$.
\end{itemize}

\noindents{-2}{\bf Contractions and $c$-contractions}. If the definition of the relation $\preccurlyeq_\phi$ is modified by omitting condition (4) and demanding that $\phi^{-1}(\star)=\emptyset$, then we deal with the contraction relation and we say that $H$ {\em is a $\phi$-generated contraction of $G$}. (Note, that condition (4) is not a requirement of the equivalence to the minor relation -- see also the proof of Lemma~\ref{tldkgn}.)
Let $c$ be a non negative integer. We say that $H$ is a {\em $c$-contraction of $G$} if $H$ is a $\phi$-generated contraction of $G$ and for all $v\in V(H)$, $G[\phi^{-1}(v)]$ is a graph of at most $c$ edges.

In this paper we we use the alternative, more complicated, definitions of minors and contractions as they are necessary for the proofs or our results.

\noindents{-2}{\bf Tree-decompositions and treewidth}. A {\em tree-decomposition} of a graph $G$, is a pair $(T,{\cal X})$, where $T$ is a tree and ${\cal X}=\{X_t:\ t\in V(T)\}$ is a family of subsets of $V(G)$, called {\em bags}, such that the following three properties are satisfied:

\begin{enumerate}
\setlength{\itemsep}{-.0pt}
\item[(1)] $\bigcup_{t\in V(T)}X_t=V(G)$,
\item[(2)] for every edge $e\in E(G)$ there exists $t\in V(T)$ such that $X_t$ contains both ends of $e$, and
\item[(3)] $\forall v\in V(G)$, the set $T_{v}=\{t\in V(T)\mid v\in X_{t}\}$ induces a tree in $T$.
\end{enumerate}

The {\em width} of a tree-decomposition is the cardinality of the maximum size bag minus 1 and the {\em treewidth} of a graph $G$ is the minimum width of a tree-decomposition of $G$. We denote the treewidth of $G$ by $\tw(G)$.  We say that a graph $H$ is a {\em partial triangulation} of a plane graph $G$
if $G$ is a spanning subgraph of $H$ and $H$ is plane.
The following result follows from~\cite{GuT12}.

\begin{proposition}
\label{propo}
Let $r$ be an integer. Then, any planar graph with treewidth at least $4.5\cdot r$ contains a partial triangulation of the $(r\times r)$-grid as a contraction.
\end{proposition}

\begin{lemma}
\label{hh4rtq}
Let $G$ be a planar graph and $k$ an integer. If $\tw(G)\geq18\cdot k$ then
$G$ contains a $(k\times k)$-grid as a distance minor.
\end{lemma}
\begin{proof}
By Proposition~\ref{propo}, the graph $G$ contains a partial triangulation $P$ of a $(4 k\times4 k)$ grid as a contraction. We claim that the $(k\times k)$-grid $L_k$ is a $\phi$-generated distance minor of $P$, where the function $\phi: E(P^{\rm \ell })\rightarrow V(L_k)\cup E(L_k)\cup\{\star\}$ is defined as follows.
For any vertex $v_{i,j}$ of $L_k$, with $i,j\in \{1,\dots,k\}$, the set $\phi^{-1}(v)$ contains the three loops on the vertices of $P$ with coordinates $(k+2i,k+2j)$, $(k+2i-1,k+2j)$, $(k+2i,k+2j-1)$ and the two edges joining these three vertices. The set $\phi^{-1}(E(L_k))$ contains all possible candidates from the edges of the underlying grid of $P$, while anything else is mapped to $\star$.

One can easily verify that $L_k$ is indeed the $(k\times k)$-grid, and that the function $\phi$ satisfies the first four conditions. It remains to show that condition (5) also holds. It suffices to show it for the case of the distance of vertices and the rest follows. Consider, for this purpose, any two vertices $v_1,v_2$ in $L_k$ and let their distance be  ${\bf dist}_{L_k}(v_{1},v_{2})=\rho\leq 2k-2$. Since any vertex in $P$ that corresponds to a vertex of $L_k$ has a distance of at least $k$ from a boundary vertex, the shortest path in $P$ containing $\phi^{-1}(v_{1})$ and $\phi^{-1}(v_{2})$ has length at least the half of the length of the shortest path in the underlying grid, i.e. $ \frac1 2\cdot2\rho\geq\rho$.
\end{proof}

%
%

\section{\bf Bidimensionality of  line intersection graphs}

\label{ldfmnfd}

Let $\bb=\{B_{1},\ldots,B_{k}\}$ be a collection of lines in the plane.
The {\em intersection graph $G_{\cal B}$} of ${\cal B}$, is a graph whose vertex set is ${\cal B}$,
and that has an edge $\{B_{i},B_{j}\}$ (for $ i\ne j$) if and only if $B_{i}$ and $B_{j}$ {\em touch},
namely $B_{i}\cap B_{j}\ne\emptyset$.\medskip

The following theorem states our main technical result.
\begin{theorem}\label{thm:bidi-polyline-intersection}
Let ${\cal B}$ be a set of lines in the plane such that for each $C_{1},C_{2}\in {\cal B}$ with $C_{1}\neq C_{2}$, the set $C_{1}\cap C_{2}$ is a finite set of points and at most two lines intersect in the same point. Let also $G_{\cal B}$ be the intersection graph of ${\cal B}$ and let $\xi=\max_{C\in {\cal B}}|C\cap \bigcup_{C'\in {\cal B}\setminus C}C'|$. Then $\tw(G_{\cal B})=O(\xi\cdot {\bf bg}(G_{\cal B}))$.
\end{theorem}

To prove Theorem~\ref{thm:bidi-polyline-intersection} we will need  a series of lemmata.
%

\begin{lemma}
\label{gkkkk5g}
Let $G$ be a graph and let $H$  be a  $c$-contraction of $G$.
 Then $\tw(G)\leq (c+1)\cdot (\tw(H)+1)-1$.
\end{lemma}
\begin{proof}
By definition, since $H$ is a $c$-contraction of $G$, there is a mapping between each vertex of $H$ and a connected set of at most $c$ edges in $G$, so that by contracting these edgesets we obtain $H$ from $G$. The endpoints of these edges form disjoint connected sets in $G$, implying a partition of the vertices of $G$ into connected sets $\{V_x\mid x\in V(H)\}$, where $|V_x|\leq c+1$ for any vertex $x\in V(H)$.


Consider now a tree decomposition  $(T,{\cal W})$ of $H$. We claim that the pair $(T,{\cal W'})$, where $W'_t:=\bigcup_{x\in W_t} V_x$ for $t\in T$ is a tree decomposition of $G$. Clearly all vertices of $G$ are included in some bag, since all vertices of $H$ did. Every edge of $G$ with both endpoints in the same part of the partition is in a bag, as each of these vertex sets is placed as a whole in the same bag. If $e$ is an edge of $G$ with endpoints in different parts of the partition, say $V_x$ and $V_y$, then this implies that $xy\in E(H)$. Thus, there is a node $t$ of $T$ for which $x,y \in W_t$ and therefore $e\in W_t'$. Moreover, the continuity property remains unaffected, since for any vertex $x \in V(H)$ all vertices in $V_x$ induce the same subtree in $T$ that $x$ did. 
\end{proof}

\begin{lemma}\label{lattgd}
Let $G$ be a graph and let $V_{1},\ldots, V_{r}$ be a partition of the vertices of $G$, such that for each $i\in\{1,\ldots,r\}$, $G[V_{i}]$ is a connected graph, and for each $i\in\{1,\ldots,r-1\}$ there exist an edge of $G$ with one endpoint in $V_{i}$ and one endpoint in $V_{i+1}$. Let also $s\in V_1$ and $t\in V_r$.
Then $G$ has a path from $s$ to $t$ with a part $P$ of length at least $\beta-\alpha+2$, where $1\leq\alpha<\beta\leq r$, so that $P$ does not contain any edge in $G[V_{i}]$ for $i \in \{1,\dots,\alpha-1\}\cup\{\beta+1,\dots,r\}$.
\end{lemma}
\begin{proof}
For each $i\in \{1,\dots,r-1\}$ let $e_i=t_i s_{i+1}$ be an edge of $G$ with one endpoint $t_i$ in $V_i$ and the other $s_{i+1}$ in $V_{i+1}$ and set $s_1=s$ and $t_r=t$.
For any $i\in \{1,\dots,r\}$, since $G[V_i]$ is a connected graph, there is a path $P_i$ from $s_i$ to $t_i$ that lies entirely in $G[V_i]$ (possibly the trivial path of no edges). Then the path $P_1e_1\dots P_{r-1}e_{r-1}P_r$ from $s$ to $t$, with its part $e_aP_{a+1}\dots e_{\beta-1}P_\beta e_\beta$ satisfies the requirements of the Lemma. \hfill \qed
\end{proof}

\begin{lemma}\label{trassnss}
Let $A$, $B$, and $C$ be graphs such that $B$ is a $\psi_{1}$-generated contraction of $A$ and $C$ is a $\psi_{2}$-generated minor of $A$ for some functions $\psi_{1}: E(A^{\rm \ell})\rightarrow V(B)\cup E(B)$ and $\psi_{2}: E(A^\ell)\rightarrow V(C)\cup E(C)\cup\{\star\}$. If 
\begin{eqnarray}
\forall_{e\in E(C)} & & |\psi_2^{-1}(e)\cap \psi_1^{-1}(E(B))|=1 \label{sdfggwr}
\end{eqnarray}
then $C$ is also a minor of $B$.
\end{lemma}

\begin{proof}
We will define a function $\phi: E(B^{\rm \ell})\rightarrow V(C)\cup E(C)\cup\{\star\}$, which in turn guaranties that $C$ is a minor of $B$. 
First of all notice that by~\eqref{sdfggwr} there is a subset $F_E\subseteq E(B)$ and a bijection $\eta$ between $F_E$ and $E(C)$. In fact, for any edge $e$ of $C$ it holds that $\eta^{-1}(e)=\psi_1(\psi_2^{-1}(e))$. We set $\phi|_{F_E}=\eta$ and we observe that $\phi$ is obliged to map any edge in $F'=E(B^\ell)\setminus F_E$ either to a vertex of $C$ or to $\star$.

Let now $v$ be a vertex of $V(C)$. Recall that $\psi_{2}^{-1}(v)$ is a non-empty solid set of edges in $A^\ell$, and thus it induces a connected subgraph, say $A_v$, in $A$. Notice also that each edge $e_v$ of $C$ incident to $v$ has a unique pre-image in $A$, which has exactly one endpoint in $A_v$. Furthermore, the graph $B_v=B[\psi_1(\psi_{2}^{-1}(v))]$ is isomorphic to the graph taken if we contract in $A_v$ all its edges that belong in $\psi^{-1}_{1}(V(B))$. As the contraction of edges does not harm the connectivity of an edgeset, it follows that $\psi_1(\psi_{2}^{-1}(v))$ is a connected set of edges in $B$ and that again $B_v$ is connected and each edge $\psi_1(\psi_{2}^{-1}(e_v))$ has an endpoint in $B_v$. We set $\phi(f)=v$, if $f\in E(B^\ell)$ is an edge of $B_v$ or a loop on a vertex of $B_v$. Let $F_V\subseteq E(B^\ell)$ be the union of all edgesets  $\phi^{-1}(v)$, for each $v \in V(C)$ and observe that $F_E\cap F_V=\emptyset$. Finally, we set $\phi(E(B^\ell)\setminus (F_E\cup F_V))=\star$. It remains to prove that the four conditions of the minor definition are satisfied.

The first and fourth conditions follow straightforwardly from the definition of $\phi$.
For the second, assume that  $v_{1},v_{2}$ are two distinct vertices of $V(C)$.
We set $N_{A}=\psi_{2}^{-1}(v_1)\cup \psi_{2}^{-1}(v_{2})$
and $N_{B}=\psi_{1}(\psi_{2}^{-1}(v_1))\cup \psi_{1}(\psi_{2}^{-1}(v_{2}))$.
Notice that $B[N_{B}]$ is isomorphic to the graph taken from $A[N_{A}]$
after we contract all  its edges that belong in $\psi^{-1}_{1}(V(B))$.
As $A[N_{A}]$ is a disconnected graph, the same holds for $B[N_{B}]$.
Therefore $N_{B}$ is disconnected.

For the third, let $e=(v_{1},v_{2})\in E(C)$. Clearly $\phi^{-1}(e)$ is not a loop since it belongs in $F_E$. Moreover, in $A$, one endpoint of $e'=\psi_{2}^{-1}(e)$
is in $A[\psi^{-1}_{2}(v_{1})]$ and the other
is in $A[\psi^{-1}_{2}(v_{2})]$. Thus, $M_{A}=\psi^{-1}_{2}(v_{1})\cup\{e'\}\cup\psi^{-1}_{2}(v_{2})$
is a connected set of $A$ while $M_{A}\setminus\{e'\}$ is not.
We set $M_{B}=\psi_{1}(\psi_{2}^{-1}(v_{1}))\cup\{e''\}\cup \psi_{1}(\psi_{2}^{-1}(v_{2}))$, where $e''=\phi^{-1}(e)$ and we observe that $B[M_{B}]$ is isomorphic to the graph taken from $A[M_{A}]$ after we contract all its edges that belong in $\psi^{-1}_{1}(V(B))$ in a way that, in this isomorphism, $e'$ corresponds to the edge $e''$.
This means that $M_{B}$ is a connected set of $B$ while $M_{B}\setminus\{e''\}$ is not. We conclude that the one endpoint of  $e''$ belongs in $\phi^{-1}(v_1)$ and the other belongs in $\phi^{-1}(v_2)$, as required.
\end{proof}

\begin{lemma}\label{biddslss}
Let $G$ be a connected  graph and  let $H$ be a $c$-contraction of $G$. If $G$ contains a $(k\times k)$-grid as a distance minor, then $H$ contains a $(k',k')$-grid as a minor, where $k'=\lfloor\frac{k-1}{2(c+1)}\rfloor+1$.
\end{lemma}

\begin{proof}
We assume that $c$ is an odd number and equivalently prove the lemma for $k'=\lfloor\frac{k-1}{2c}\rfloor+1$.

Let $H$ be a $\sigma$-generated contraction of $G$ for some $\sigma: E(G^\ell)\rightarrow V(H)\cup E(H)$ such that $G[\sigma^{-1}(v)]$ is a graph of at most $c$ edges for all $v\in V(H)$. Suppose also that $G$ contains a $(k\times k)$-grid $L_{k}$ as a distance minor via a function
$\phi: E(G^\ell)\rightarrow V(L_{k})\cup E(L_{k})\cup\{\star\}$.

We assume that $V(L_{k})=\{1,\ldots,k\}^{2}$ where each $(i,j)$
corresponds to its grid  coordinates.
%
%
Our target is to prove that the $(k'\times k')$ grid $L_{k'}$ is a minor of $H$.
We define $\alpha: \{1,\ldots,k'\}\rightarrow \{1,\ldots,k\}$ such that $\alpha(i)=2(i-1)c+1$. Notice that this definition is possible as $2(k'-1)c+1\leq k$.
For each $(i,j)\in\{1,\ldots,k'\}^{2}$, we define a {\em horizontal} and a {\em vertical}  set of vertices in $V(L_k)$,

\begin{figure}[ht]
\begin{center}
\scalebox{.78}{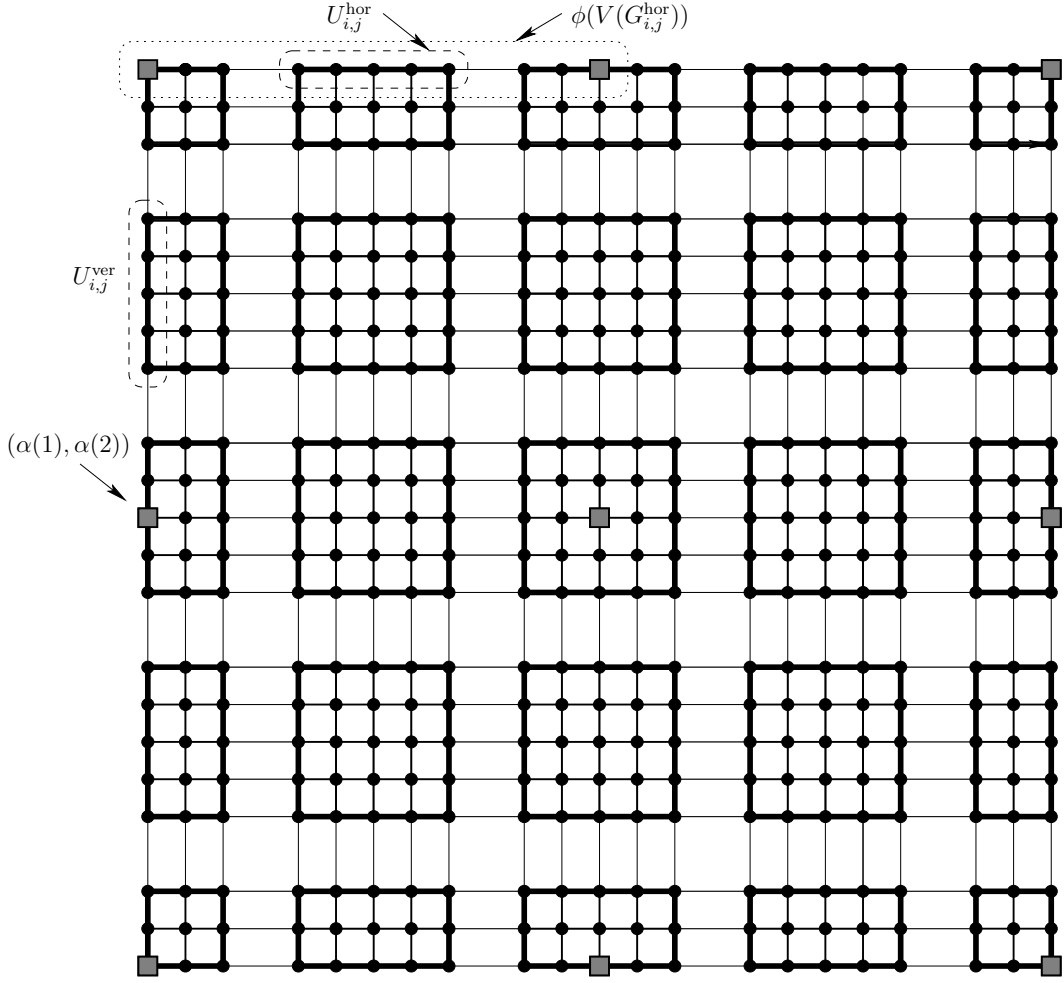}
\end{center}
\caption{An example of the proof of Lemma~\ref{biddslss} for $c=5$, $k=21$, and $k'=3$.}
\label{fig:this2one}
\end{figure}

\begin{eqnarray}
U^{\rm hor}_{i,j} & = & \bigcup_{r\in\{\alpha(i)+(c+1)/2,\ldots,\alpha(i+1)-(c+1)/2\}}(r,\alpha(j)),\label{tlpd}\\
U^{\rm ver}_{i,j} & = & \bigcup_{r\in\{\alpha(j)+(c+1)/2,\ldots,\alpha(j+1)-(c+1)/2\}}(\alpha(i),r)\label{fdsf}
\end{eqnarray}
and let ${\cal U}$ be the collection of all sets $U_{i,j}^{\rm hor}$ or $U_{i,j}^{\rm ver}$ defined in~\eqref{tlpd} and~\eqref{fdsf}.
For every horizontal (resp. vertical) $U\in {\cal U}$, we denote by ${\cal E}(U) \subseteq E(L_k)$ the set containing all horizontal (resp. vertical) edges of $L_{k}$ with an endpoint in $U$.
We will prove the following claim:
\begin{itemize}
\item[($*$)] {\em Let $U_{1}$ and $U_{2}$ be two different sets  of ${\cal U}$ and let $e_1,e_2$ be two edges of $G$ such that $\phi(e_i)\in{\cal E}(U_{i})\cup U_{i}$, for $i=\{1,2\}$. Then, there are no disjoint paths of length at most $c$ from the endpoints of $e_1$ to the endpoints of $e_2$ in $G$.}
\end{itemize}

\noindent Since $L_{k}$ is a distance minor of $G$, it suffices to show that there is no cycle in $L_k$, that contains $\phi(e_1)$ and $ \phi(e_2)$ together with two paths between them of length at most $c$. Let us suppose that such a cycle exists. Notice that by the definition of ${\cal U}$, if two vertices $x,y$ of $V(L_{k})$ belong to two different sets of ${\cal U}$, then ${\bf dist}_{L_{k}}(x,y)\geq c+1$. This implies that $\phi(e_i)$ must be an edge $v_iu_i$ of $L_k$ with only one endpoint, say $u_i$, in $U_i$, for $i=\{1,2\}$, or else we are done. Likewise, it holds that ${\bf dist}_{L_{k}}(u_1,u_2)\geq c+1$ and hence one path of length at most $c$ of the cycle must be from $v_1$ to $u_2$, the other from $v_2$ to $u_1$. It follows, that the edges $v_1u_1$ and $v_2u_2$ cannot be both vertical nor horizontal, and all vertices of the two disjoint paths lie inside the square part of the grid these two edges define. This contradicts the planarity of the grid, which completes the proof of the claim.

For every $i,j\in\{1,\ldots,k'\}^{2}$ we choose arbitrarily a vertex $v_{i,j}$ from the graph $G[\phi^{-1}(\alpha(i),\alpha(j))]$.
This selection creates a collection of $k'\times k'$ vertices of $G$.

For each pair $\{(i,j),(i+1,j)\}$ where $(i,j)\in \{1,\ldots,k'-1\}\times\{1,\ldots,k'\}$,
we observe  that the graph $$G^{\rm hor}_{i,j}=G[\bigcup_{i'\in\{\alpha(i),\ldots,\alpha(i+1)\}}\phi^{-1}(i',\alpha(j))]$$
is connected, and for every $i'\in\{\alpha(i),\ldots,\alpha(i+1)\}$ the sets $\phi^{-1}(i',\alpha(j))$ form a partition of $V(G^{\rm hor}_{i,j})$ and
there is an edge of $G$ between $\phi^{-1}(i',\alpha(j))$ and $\phi^{-1}(i'+1,\alpha(j))$.
From Lemma~\ref{lattgd}, $G^{\rm hor}_{i,j}$ contains
a path $P_{i,j}^{\rm hor}$ from $v_{i,j}$ to $v_{i+1,j}$ with a part of length at least $c+1$ in $\phi^{-1}({\cal E}(U_{i,j}^{\rm hor}))\cup\phi^{-1}(U_{i,j}^{\rm hor})$.
Clearly, one of the edges in this part of the path, say  $e_{i,j}^{\rm hor}$, is an edge  of $\sigma^{-1}(E(H))$.
We denote by $\stackrel{\rightarrow}{P}_{i,j}^{}$ (resp. $\stackrel{\leftarrow}{P}_{i+1,y}^{}$) the part of $P_{i,j}^{\rm hor}$
starting from  $v_{i,y}$  (resp.  $v_{i+1,y}$) and containing only one endpoint of $e_{i,j}^{\rm hor}$.

Working in the same way as before but following the ``vertical'' instead of ``horizontal'' direction,
for each pair $\{(i,j),(i,j+1)\}$ where $(i,j)\in \{1,\ldots,k'\}\times\{1,\ldots,k'-1\}$,
we define the graph
 $$G^{\rm ver}_{i,j}=G[\bigcup_{j'\in\{\alpha(j),\ldots,\alpha(j+1)\}}\phi^{-1}(\alpha(i),j')]$$
and we find the path $P_{i,j}^{\rm ver}$ in it starting from $v_{i,j}$ finishing in $v_{i,j+1}$
and containing an edge $e_{i,j}^{\rm ver}$ of $\sigma^{-1}(E(H))$
that belongs in $\phi^{-1}({\cal E}(U_{i,j}^{\rm ver}))\cup\phi^{-1}(U_{i,j}^{\rm ver})$.
As before, $P_{i,j}^{\rm ver}$ is decomposed to a path ${\downarrow}{P}_{i,j}^{}$ (containing $v_{i,j}$),
the edge $e_{i,j}^{\rm ver}$, and the path
${\uparrow}{P}_{i,j+1}^{}$ (containing $v_{i,j+1}$).
Let, finally, $E^{*}$ be the set containing each $e_{x,y}^{\rm hor}$
and each $e_{x,y}^{\rm ver}$.


From Lemma~\ref{trassnss}, to prove that $L_{k'}$ is a minor of $H$,
it is enough to
define a function $\tau: E(G^{\rm \ell})\rightarrow  V(L_{k'})\cup E(L_{k'})\cup\{\star\}$
certifying that $L_{k'}$ is a minor of $G$ in a way that
$\forall {f\in E(L_{k'})}\ |\tau^{-1}(f)\cap \sigma^{-1}(E(H)|=1$.
For this, for every $(x,y)\in\{1,\ldots,k'\}$
we define $E_{x,y}$ as the union of the edges and the loops of the vertices of every path that exists in the set
 $\{\stackrel{\leftarrow}{P}_{x,y}^{},\stackrel{\rightarrow}{P}_{x,y}^{},{\downarrow}{P}_{x,y}^{},{\uparrow}{P}_{x,y}^{}\}$
 and for each $e\in E_{x,y}$ we set $\tau(e)=(x,y)$.
 Notice that for every $(x,y)\in\{1,\ldots,k'\}^{2}$, $G[\tau^{-1}(x,y)]$ is the union of a set of paths of $G$ having
 a vertex in common, thus it
 induces a connected subgraph of $G$. Let now $e$ be an edge of $L_{k'}$.
 In case  $e=\{(x,y),(x+1,y)\}$ (resp. $e=\{(x,y),(x,y+1)\}$), then, by its definition,
 the edge $e_{x,y}^{\rm hor}$ (reps.  $e_{x,y}^{\rm ver}$) connects an endpoint $v_{1}$ of an edge in $\tau^{-1}(x,y)$ (resp. $\tau^{-1}(x,y)$) with an endpoint
 $v_{2}$ of an edge in $\tau^{-1}(x+1,y)$ (resp. $\tau^{-1}(x,y+1)$). In any case, we set
 $\tau(v_{1}v_{2})=e$. It follows that $\tau(E^{*})=E(L_{k'})$.
For all edges of $G^\ell$ whose image has not been defined so far, we set $\tau(e)=\star$. It is now easy to verify that $\tau$ is a well-defined function and that $L_{k'}$ is a $\tau$-generated minor of  $G$.


Next we prove that $\forall {f\in E(L_{k'})}\ |\tau^{-1}(f)\cap \sigma^{-1}(E(H)|=1$. For this, first of all notice that, by the definition of $\tau$,
all edges in $\tau^{-1}(E(L_{k'}))=E^*$ are edges of $\sigma^{-1}(E(H))$.
Therefore, it suffices to prove that for each $e\in E(H)$,
$\sigma^{-1}(e)$ contains no more than one edge from $E^*$.
Suppose in contrary that $e_{1},e_{2}\in \sigma^{-1}(e)\cap E^*$
and $e_{1}\neq e_{2}$.
As $\sigma(e_{1})=\sigma(e_{2})=e$, it follows that each $e_{i}$ has an endpoint $w_{i}$
in $\sigma^{-1}(w)$ and an endpoint $z_i$ in $\sigma^{-1}(z)$, where $wz=e$. Since each subgraph $G[\sigma^{-1}(w)]$ and $G[\sigma^{-1}(z)]$ is connected, has at most $c$ edges and both are disjoint, there are two disjoint paths of length at most $c$ in $G$ from $w_1$ to $w_2$ and from $z_1$ to $z_2$, a contradiction to ($*$) as $e_1,e_2\in E^*$.
\end{proof}

\begin{lemma}\label{pecm25r}
Let $H_{1}$ and $H_{2}$ be two graphs. Consider a graph $G$ such that $H_1$ is a $c_{1}$-contraction of $G$ and $H_2$ is a $c_{2}$-contraction of $G$. If $H_1$ is planar then $\tw(H_{2})=O(c_{1}\cdot c_{2}\cdot {\bf bg}(H_{2}))=36\cdot (c_{1}+1)\cdot (c_{2}+1)\cdot [{\bf bg}(H_{2})-1]+O(c_{1})$.
\end{lemma}
\begin{proof} Let $H_1,H_2$ be two contractions of $G$ generated by some $\sigma_{i}: E(G^\ell)\rightarrow V(H_i)\cup E(H_i),\ i=1,2$ respectively. Let $r=\tw(H_{2})$. As $G$ contains $H_{2}$ as a contraction, it follows that $\tw(G)\geq r$.
By Lemma~\ref{gkkkk5g}, $\tw(H_{1})\geq (r+1)/(c_{1}+1)-1$.
Since $H_{1}$ is planar, by Lemma~\ref{hh4rtq},  $H_{1}$ contains $L_{r'}$ as a distance minor, where $r'=\lfloor \frac{1}{18}\cdot( \frac{r+1}{c_{1}+1}-1)\rfloor$.
As $H_1$ is a contraction of $G$, then also $G$ contains $L_{r'}$ as a distance minor. By Lemma~\ref{biddslss}, $H_{2}$ contains as a minor an $(r'',r'')$-grid, where $r''=\lfloor\frac{r'-1}{2(c_2+1)}\rfloor+1$, as claimed. \end{proof}

\begin{proof}[{\it Proof of Theorem~\ref{thm:bidi-polyline-intersection}.}]
Given a planar drawing of the lines of ${\cal B}$, consider any crossing $p$, that is a point of the plane that belongs to more than one line. By the assumptions, $p$ belongs to exactly two lines, say $L_1,L_2$ and there is an open disc $D$ of the plane containing $p$, but no lines other than $L_1$ and $L_2$ and no other point that belongs to more than one line. In addition, w.l.o.g. we can always assume that $p$ is not an endpoint of $L_1$ or $L_2$; or else we can stretch inside $D$ the line that ends in $p$ without further altering the setting.

Then, let $G_1$ be the simple graph with an embedding in the plane, in which all endpoints of lines in $\bb$ are vertices of $G_1$ and every line $L\in \bb$ is an edge of $G_1$ joining the two vertices, which are endpoints of $L$. Note that the graph $G_1$ is not necessarily planar -- in fact, any crossing of two lines in $\bb$ is as well a crossing of the corresponding edges of $G_1$ in this embedding.

For every crossing $p$ of two lines $L_1,L_2$ in $\bb$ and hence of the corresponding edges $e_1,e_2$ of $G_1$, we can consider as above an open disc $D$ of the plane in a way, so that $D\cap e=\emptyset$ for any edge $e \in E(G_1)\setminus\{e_1,e_2\}$, the only point in $D$ that belongs to both edges is $p$, no vertex of $G_1$ lies in $D$ and all considered discs are pairwise disjoint. Then, we subdivide $e_1$ and $e_2$ by adding two new vertices $x,y$ in $D\setminus\{p\}$ and we join $x$ and $y$ with a new edge $f$ that lies entirely in the disc $D$ and meets $L_1$ and $L_2$ only at its endpoints. We denote as $M$ the set of these new edges. Notice that we can contract the edge $f$ inside the disc $D$ so that the resulting vertex is the point $p$, leaving the embedding of the graph outside $D$ untouched. By doing so for every edge in $M$, we obtain a planar embedding of a graph. Let $H$ be this graph and let $G$ be the graph before contracting the edges in $M$, i.e. $G/M=H$. Clearly, $H$ is an $1$-contraction of $G$. Moreover, if we contract all edges of $G$ that are not in $M$, we obtain the intersection graph $G_{\cal B}$. Since every edge of $G_1$ was subdivided into at most $\xi+1$ edges of $G$, the graph $G_{\cal B}$ is a $(\xi+1)$-contraction of $G$ and the result follows from Lemma~\ref{pecm25r}.
\end{proof}

\section{\bf Modeling body intersections by intersection of polysegments}

\label{kfofhfm}

Let $\bb=\{B_{1},\ldots,B_{k}\}$ be a collection of 2-dimensional geometric bodies in the plane. We assume that if two bodies do intersect each other, then every connected component of the intersection has a non-empty interior. 
 Our goal is to associate each geometric body $B\in\bb$ with a polysegment $C$ such that the resulting set $\bb'$ of polysegments conveys all necessary information regarding the disposition of the bodies in the plane and their intersections.

For every body $B_i$ let us pick a point of the sphere $p_i$ that lies in $B_i$ and for every body $B_j$ touching $B_i$ ($i\ne j$) in $\bb$, a point $p_{ij}$ that lies in $B_{i}\cap B_{j}$. We can assume without loss of generality, that these points are pairwise distinct and that any three of them are not co-linear. We stress that since $\bb$ is finite this assumption is safe, because we can always consider an open disc $D_p$ of small radius around any given point $p$ of the sphere, such that if $p$ lies in a body $B$ then we can replace $B$ with the possibly expanded body $B'=B\cup D_p$ without altering the intersection graph of $\bb$. Let now ${\cal P}_i$ be the set of all points that contain the index $i$ in the assigned subscript. 
A geometric body $B$ of the sphere is {\em $\rho$-convex}, if for any two points of $B$ there exists a polysegment of length $\rho$ that lies entirely inside $B$ and its endpoints are the given two points. Notice, that the definition of a $\rho$-convex body naturally extends the standard definition of a convex body, which under this new perspective is also called $1$-convex. 

\begin{lemma}\label{lemma:r-convex-polyline}
For any collection of $\rho$-convex bodies $\bb$ on the sphere, there exists a collection of polysegments $\bb'$ and a bijection $\phi:\ \bb\rightarrow\bb'$, such that 
two bodies in $\bb$ touch if and only if the corresponding polysegments in $\bb'$ touch.
Moreover, each polysegment $C\in \bb'$ is crossed by the polysegments from $\bb'\setminus C$ at most $\xi=O(\rho^2\Delta^3)$ times, where $\Delta$ is the maximum degree in the intersection graph $G_{\bb}$ of $\bb$.
\end{lemma}

\begin{proof}
Let us then consider such a collection $\bb$ of $\rho$-convex bodies $\{B_{1},\ldots,B_{k}\}$ where any body in $\bb$ touches at most $\Delta$ other bodies, for some positive integers $\rho$ and $\Delta$. 
Clearly $|{\cal P}_i|\leq\Delta+1$, for $1\leq i\leq k$. 
Notice that as a body $B_i$ of $\bb$ is $\rho$-convex, for any two points in ${\cal P}_i$ there is a polysegment of length at most $\rho$ that joins them. We create the following drawing of a polysegment. 

Pick two points in ${\cal P}_i$ and join them with a polysegment that lies in $B_i$. As long as there are still points in ${\cal P}_i$ we did not join, pick one and join it to a ${\cal P}_i$ point or bend point of the so far drawing, so that the new polysegment does not cross the drawing. To see this is possible, consider that if the polysegment crosses a straight line segment of the drawing so far, we can simply replace it by one that joins the new point of ${\cal P}_i$ to a bend point that lies on the border of the straight line segment in question. 
In the end we obtain the drawing of a tree on the sphere containing all points in ${\cal P}_i$. 
At each step we joined at least one new point in ${\cal P}_i$, and in doing so we introduced at most $\rho-1$ new bend points. 
It follows that the drawing has in total at most $\rho\cdot\Delta+1$ bend and $P_i$ points. 

Finally, we circumscribe around the drawing of this tree a polysegment $C_i$ which contains all points in ${\cal P}_i$,
so that to every straight line segment of the tree correspond two straight line segments of $C_i$, and any bend point of $C_i$ lies in an open disc of small radius around the corresponding bend point of the drawing of the tree.
Clearly the polysegment $C_i$ has length at most $2\rho\cdot\Delta$.

We have thus showed, that for any $\rho$-convex body $B_i\in \bb$, there exists a polysegment $C_i$ of length $O(\rho\cdot\Delta)$, that lies entirely inside $B_i$ and contains all points in ${\cal P}_i$. The mapping of a body $B_i$ to a polysegment $C_i$ defines a bijection $\phi$ between the elements of the collection $\bb$ and the collection of polysegments $\bb'=\{C_i:1\leq i\leq k\}$. 
By construction, two distinct polysegments $C_i,C_j$ both contain the point $p_{ij}$ and hence they do share a point, if the corresponding bodies $B_i,B_j$ touch.
On the other hand, as the polysegments lie inside the geometrical bodies, the polysegments $C_i,C_j$ share a point of the sphere only if the associated bodies touch. This also implies that a polysegment $C_i$ shares at most $O((\rho\cdot\Delta)^2\cdot\Delta)=O(\rho^2\cdot\Delta^3)$ points of the sphere with other polysegments in $\bb'$, concluding the proof. 
\end{proof}

Straightforwardly applying Theorem~\ref{thm:bidi-polyline-intersection} to the sets of polysegments constructed in Lemma~\ref{lemma:r-convex-polyline}, results to the following theorem for $\rho$-convex geometric bodies.

\begin{theorem}\label{thm:r-convex}
Let ${\cal B}$ be a set of $\rho$-convex bodies such that for each $B_{1},B_{2}\in {\cal B}$ with $B_{1}\cap B_{2}\neq\emptyset$, the set $B_{1}\cap B_{2}$ has non-empty interior. Let $G_{\cal B}$ be the intersection graph of ${\cal B}$ and let $\Delta$ be the maximum degree of $G_{\cal B}$. Then $\tw(G_{\cal B})=O(\rho^2\Delta^3\cdot {\bf bg}(G_{\cal B}))$.
\end{theorem}

Given a positive real number $\alpha $, we define the class of 
{\em $\alpha $-fat convex intersection graphs} as the class 
containing an intersection graph $G_{\cal B}$ of a collection $\bb$ of convex bodies, 
if the ratio between the maximum and the minimum radius of a circle 
where all objects in ${\cal B}$ can be circumscribed, and inscribed
respectivelly,  is upper bounded by $\alpha$.
The following lemma describes the manner in which the convex bodies of such a collection are being modeled by polysegments. 

\begin{lemma}\label{lemma:fat-convex-polyline}
Let $H$ be a graph on $h$ vertices and let $\bb$ be a collection of convex bodies on the sphere. If the intersection graph of $\bb$ is $\alpha$-fat and does not contain graph $H$ as a subgraph, then there exists a collection of polysegments $\bb'$ and a bijection $\phi:\ \bb\rightarrow\bb'$ 
such that two bodies in $\bb$ touch if and only if the corresponding polysegments in $\bb'$ touch.
Moreover, each polysegment $C\in \bb'$ is crossed by the polysegments from $\bb'\setminus C$ at most $\xi=O(\alpha^6\cdot h^3)$ times.
\end{lemma}

\begin{proof}
Let us then consider such a collection $\bb$ of convex bodies $\{B_{1},\ldots,B_{k}\}$ where the intersection graph $G_\bb$ is $\alpha$-fat and has maximum degree $\Delta$, for a positive integer $\Delta$ and a positive real $\alpha$. 
It follows that for a set of points ${\cal P}_i$ corresponding to a body $B_i$ of $\bb$, it holds that  $|{\cal P}_i|\leq\Delta+1$. As the body $B_i$ is convex, there exists a polysegment $C_i$ of length $|{\cal P}_i|-1$ containing all points in ${\cal P}_i$, that lies entirely in $B_i$. Let $\bb'$ be the collection of the polysegments $\{C_i:1\leq i\leq k\}$. Clearly, two distinct polysegments of $\bb'$ share a point of the sphere if and only if the corresponding bodies touch. Furthermore, a polysegment has at most $O(\Delta^3)$ common points with other polsegments in $\bb'$. 

We claim that $\Delta\leq16\alpha^2\cdot h$, which directly implies the bound of the lemma.
To contradiction, assume that $G_\bb$ has a vertex with degree more than $16h\alpha^2$. Then, for the corresponding body $B$ of the collection $\bb$, pick an arbitrary point $x\in B$. Since each object in $\mathcal{B}$ can be inscribed in a circle of radius $R$, all bodies intersecting $B$ belong to the ball $C$ of radius $4R$ with the center in $x$. This ball has area $16{\rm {\rm \pi}} R^2$. Since each object in $\mathcal{B}$ contains a circle of radius $r$, the sum of the areas of the bodies intersecting $B$ exceeds $16h\alpha^2\times {\rm {\rm \pi}} r^2=16 {\rm {\rm \pi}} h R^2$. Therefore, there is a point in the ball $C$ belonging to at least $h+1$ bodies from $\mathcal{B}$. Thus, $G_\mathcal{B}$ contains a clique of size $h+1$ that contradicts to the assumption that it is $H$-free.
\end{proof}

Again, by straightforwardly applying Theorem~\ref{thm:bidi-polyline-intersection} to the sets of polysegments constructed in Lemma~\ref{lemma:fat-convex-polyline}, we derive an improved theorem for $H$-free $\alpha$-fat convex intersection graphs of geometric bodies.
\begin{theorem}\label{thm:a-fat}
Let $H$ be a graph on $h$ vertices, and let $\bb$ be a collection of convex bodies on the sphere such that for each $B_{1},B_{2}\in {\cal B}$ with $B_{1}\cap B_{2}\neq\emptyset$, the set $B_{1}\cap B_{2}$ has a non-empty interior. If the intersection graph $G_{\bb}$ of $\bb$ is $\alpha$-fat and does not contain $H$ as a subgraph, then $\tw(G_{\cal B})=O(\alpha^6 h^3\cdot {\bf bg}(G_{\cal B}))$.
\end{theorem}

\section{Conclusions and further research}

\label{kdhufkl}

We believe that the applicability of our combinatorial 
results is even wider than what is explained in the previous section.
The main combinatorial engine of this paper is Lemma~\ref{pecm25r}
that essentially induces an edit-distance notion between graphs under contractibility. 
This is materialized by the following definition.

\begin{definition}
Let $G_{1}$ and $G_{2}$ be graphs. We define the 
{\em contraction-edit distance between} $G_{1}$ and $G_{2}$,
denoted by ${\bf cdist}(G_{1},G_{2})$, 
as the minimum
$c$ for which there exists a graph that contains 
both $G_{1}$ and $G_{2}$ as $c$-contractions. 
Given a graph $G$ we define ${\cal B}^{c}(G)=\{H\mid {\bf cdist}(G,H)\leq c\}$.
Finally, given a graph class ${\cal G}$, we define ${\cal B}^{c}({\cal G})=\bigcup_{G\in {\cal G}}{\cal B}^{c}(G)$.
We refer to the class ${\cal B}^{c}({\cal G})$ as the {\cal $c$-contraction extension} of the class ${\cal G}$.
\end{definition}
%
%
%
%
%

A direct consequence of Lemma~\ref{pecm25r} is the following:

\begin{corollary}
\label{lk3ers5g67}
Let ${\cal P}$ be the class of planar graphs, then for every fixed constant $c$, 
${\cal B}^{c}({\cal P})$ satisfies~\eqref{gjjfkr43d}.
\end{corollary}

Actually, Corollary~\ref{lk3ers5g67} can be extended much further than planar graphs.
For this, the only we need analogues of Lemma~\ref{hh4rtq} for more general 
graph classes. Using the main result of~\cite{FominGT11cont}, it follows 
that Lemma~\ref{hh4rtq} is qualitatively correct for every graph class that excludes an apex graph as a minor
(an apex graph is a graph that can become planar after the removal of a vertex).
By plugging this more general version of  Lemma~\ref{hh4rtq}  to the proofs of the previous 
section we obtain the following.

\begin{theorem}
\label{lk3ers5g}
Let $H$ be an apex-minor free graph and let ${\cal G}_{H}$ be the class of  graphs
excluding $H$ as a minor. Then, for every fixed constant $c$, the class 
${\cal B}^{c}({\cal G}_{H})$ satisfies~\eqref{gjjfkr43d}.
\end{theorem}

All the algorithmic 
applications of this paper follow by the fact that all geometric intersection 
graph classes 
considered in this paper are subsets of ${\cal B}^{c}({\cal P})$ for some 
choice of $c$. Clearly, Theorem~\ref{lk3ers5g} offers a much more wide framework for this,
including graphs of bounded genus (including intersection graphs of lines or polygons 
on surfaces), graphs excluding a single-crossing graph, and 
$K_{3,r}$-minor free graphs.
We believe  that Theorem~\ref{lk3ers5g}, that is the 
most general combinatorial extension of our results may have 
applications to more general combinatorial objects than just intersection 
graph classes. We leave this question open for further research.

\subsection*{Acknowledgements}  The authors wish to thank all the anonymous reviewers for 
their valuable comments on previous versions of this paper.

%
%

\end{document}

%% file: gridu7j.pdf_t
\begin{picture}(0,0)%
\includegraphics{gridu7j.pdf}%
\end{picture}%
\setlength{\unitlength}{3947sp}%
\begingroup\makeatletter\ifx\SetFigFont\undefined%
\gdef\SetFigFont#1#2#3#4#5{%
  \reset@font\fontsize{#1}{#2pt}%
  \fontfamily{#3}\fontseries{#4}\fontshape{#5}%
  \selectfont}%
\fi\endgroup%
\begin{picture}(8437,7855)(8161,75742)
\put(8701,81314){\makebox(0,0)[lb]{\smash{{\SetFigFont{12}{14.4}{\rmdefault}{\mddefault}{\updefault}{\color[rgb]{0,0,0}$U_{i,j}^{\rm ver}$}%
}}}}
\put(12676,83414){\makebox(0,0)[lb]{\smash{{\SetFigFont{12}{14.4}{\rmdefault}{\mddefault}{\updefault}{\color[rgb]{0,0,0}$\phi(V(G_{i,j}^{\rm hor}))$}%
}}}}
\put(8176,79964){\makebox(0,0)[lb]{\smash{{\SetFigFont{12}{14.4}{\rmdefault}{\mddefault}{\updefault}{\color[rgb]{0,0,0}$(\alpha(1),\alpha(2))$}%
}}}}
\put(10726,83414){\makebox(0,0)[lb]{\smash{{\SetFigFont{12}{14.4}{\rmdefault}{\mddefault}{\updefault}{\color[rgb]{0,0,0}$U_{i,j}^{\rm hor}$}%
}}}}
\end{picture}%